\documentclass[11pt]{article}

\input{style}

\newcommand*{\tv}{{\mathrm{TV}}}
\newcommand*{\Mc}{{\mathcal{M}}}
\newcommand*{\Sc}{{\mathcal{S}}}
\newcommand*{\Prob}{{\mathbb{P}}}
\let\O\relax
\newcommand*{\O}{{\mathcal{O}}}

\newcommand*{\tmix}{t_{\mathrm{mix}}}
\declarenameddelims{D}{{\mathcal{D}_{\mathrm{KL}}}}\lparen\rparen

\title{Isotropy and Log-Concave Polynomials: Accelerated Sampling and High-Precision Counting of Matroid Bases}

\author{Nima Anari}
\affil{\small Stanford University, \textsf{anari@cs.stanford.edu}}
\declareperson{Nima}

\author{Micha{\l} Derezi\'{n}ski}
\affil{\small University of California, Berkeley, \textsf{mderezin@berkeley.edu}}
\declareperson{Michal}

\bibliography{refs}

\begin{document}
	\maketitle
	
	\begin{abstract}
	We define a notion of isotropy for discrete set distributions. If $\mu$ is a distribution over subsets $S$ of a ground set $[n]$, we say that $\mu$ is in isotropic position if $\PrX{S\sim \mu}{e\in S}$ is the same for all $e\in [n]$. We design a new approximate sampling algorithm that leverages isotropy for the class of distributions $\mu$ that have a log-concave generating polynomial; this class includes determinantal point processes, strongly Rayleigh distributions, and uniform distributions over matroid bases. We show that when $\mu$ is in approximately isotropic position, the running time of our algorithm depends polynomially on the size of the set $S$, and only logarithmically on $n$. When $n$ is much larger than the size of $S$, this is significantly faster than prior algorithms, and can even be sublinear in $n$. We then show how to transform a non-isotropic $\mu$ into an equivalent approximately isotropic form with a polynomial-time preprocessing step, accelerating subsequent sampling times. The main new ingredient enabling our algorithms is a class of negative dependence inequalities that may be of independent interest.
	
	 As an application of our results, we show how to approximately count bases of a matroid of rank $k$ over a ground set of $n$ elements to within a factor of $1+\epsilon$ in time $ O((n+1/\epsilon^2)\cdot \poly(k, \log n))$. This is the first algorithm that runs in nearly linear time for fixed rank $k$, and achieves an inverse polynomially low approximation error.
	\end{abstract}
	
	\section{Introduction}
\label{sec:intro}


Designing efficient algorithms for sampling from continuous distributions of convex type, i.e., those with a log-concave density, has a long history \cite[see][ for relevant surveys]{Vem05, Vem10, LV18}. A challenging part of many algorithms for this problem has been transforming the distribution via a linear map to an equivalent standard form called the  \emph{isotropic position}. To date, the running times for the fastest algorithms for sampling from convex polytopes and log-concave distributions are dominated by finding the correct scaling linear transform \cite{LV06, MV19}. As our main contribution, we introduce an analogous notion of isotropic position for discrete distributions, and design sampling algorithms that can take advantage of the isotropic position for distributions with a log-concave generating polynomial; a class that should in many ways be thought of as the discrete analog of log-concave distributions.

We study the problem of sampling from a distribution on size $k$ subsets of $[n]=\set{1,\dots,n}$ given by a density function $\mu:\binom{[n]}{k}\to \R_{\geq 0}$.\footnote{We note that the choice of $\binom{[n]}{k}$ as the domain is simply a standard form and trivial transformations can be applied to many high-dimensional discrete distributions to obtain this form \cite[see, e.g.,][]{ALO20}.} The task of approximate sampling is to query the function $\mu$ repeatedly in order to produce a random set $S\in \binom{[n]}{k}$, such that approximately $\Pr{S}\propto \mu(S)$. While there is no hope for efficient algorithms that run in nontrivial time ($\ll n^k$) for general $\mu$, recent works have identified a tractable class of distributions with many nice structural properties \cite{AOV18, ALOV19, BH19, CGM19, ALOV20}, namely, the class of distributions with a log-concave generating polynomial. This class consists of all $\mu$ where the following polynomial is log-concave as a function over $\R_{\geq 0}^n$:
\[ g_\mu(z_1,\dots,z_n):=\sum_{S\in \binom{[n]}{k}} \mu(S) \prod_{i\in S} z_i. \]
For any fixed $n$ and $k$, this class has nonempty interior in the space of all distributions on $\binom{[n]}{k}$. But more importantly, it includes well-studied distributions such as the uniform distribution over bases and/or independent sets of a matroid \cite{ALOV19, BH19}, strongly Rayleigh measures and specifically determinantal point processes \cite{BBL09}. For more examples refer to \cite{ALOV19, BH19}. Approximate sampling and counting in matroids has connections to many natural combinatorial problems, which motivates the need for finding fast algorithms for these problems. For example, estimating network reliability can be cast as counting independent sets in a cographic matroid \cite{CC97} and estimating reliability of a linear code against erasures can be cast as counting independent sets in a linear matroid \cite{Cam98}. For an exposition on similar reliability quantities of rigidity matroids see \cite{Gra91}.

\Textcite{ALOV19} used natural random walks studied in the context of high-dimensional expanders \cite{KM16, DK17, KO20} to show that distributions $\mu$ with a log-concave generating polynomial can be approximately sampled in polynomial time. Specifically, they showed that the following ``down-up'' random walk can be used to sample from $\mu$:
\begin{itemize}
	\item Starting from a set $S_0\in \binom{[n]}{k}$, repeat for $i=0, \dots, t-1$:
	\begin{itemize}
		\item Sample $e\in S_i$ uniformly at random and let $T_i=S_i-\set{e}$.
		\item From all $f\in [n]-T_i$ pick one with probability $\propto \mu(T_i\cup \set{f})$ and let $S_{i+1}=T_i\cup \set{f}$.
	\end{itemize}
	\item Output $S_t$.
\end{itemize}
Our understanding of the mixing time of this random walk has gradually improved to a tight bound in a series of works \cite{ALOV19, CGM19, ALOV20}. We now know that after $t=O(k\log(k/\epsilon))$ steps the distribution of $S_t$ becomes $\epsilon$-close to $\mu$ in total variation distance \cite{ALOV20}. However, each step of the random walk requires $O(n)$ evaluations of the density $\mu$, which brings the total complexity to $O(nk\log(k/\epsilon))$, much worse than nearly-linear in $k$. This problem is exacerbated when sampling is used to solve counting \cite{JVV86}, that is approximating the partition function $\sum_{S}\mu(S)$. Known reductions from approximate counting to approximate sampling \cite{JVV86} tack on at least an additional multiplicative factor of $1/\epsilon^2$, where $\epsilon$ is the desired relative error, yielding a running time that grows at least as quickly as $\frac{n}{\epsilon^2} \poly(k, \log n)$. This is a barrier against using these sampling and counting algorithms. 

The dependence of the running time on $n$ can be prohibitive in some natural applications where $k$ is of moderate size but $n$ is very large or even potentially $\infty$ in the case of continuous determinantal point processes \cite{AFT13}. Other variants of this random walk have been studied for subclasses of log-concave polynomials \cite{HS19}, but they too have a similar total running time.

 The starting point of this work are the following questions:
\begin{quote}
	When can the dependence on $n$ be avoided? Can we sample from $\mu$ in time $\poly(k)$?
\end{quote}
A natural barrier to a positive answer is the existence of \emph{important elements} in the ground set $[n]$. Consider a distribution defined by $\mu(S)\propto\prod_{i\in S}\lambda_i$, where $\lambda_1,\dots,\lambda_n\in \R_{\geq 0}$. If one of the $\lambda_i$s is much larger than the others, we should return a set that with high probability contains $i$. However no algorithm can find this \emph{important} $i$ with fewer than $n/k$ queries to $\mu$. Our main result shows that in a sense, identifying the important elements is the only barrier.

To state our main result, it is convenient to assume that we have an oracle $\O$ that can produce i.i.d.\ samples of a fixed distribution on $[n]$. One should think of the oracle as trying to sample approximately proportional to the marginals of the target distribution $\mu:\binom{[n]}{k}\to \R_{\geq 0}$. Note that the sum of marginals $\sum_{i} \PrX{S\sim \mu}{i\in S}$ is always $k$ for a distribution $\mu:\binom{[n]}{k}\to \R_{\geq 0}$.
\begin{theorem}\label{thm:main}
	Given oracle access to a density function $\mu:\binom{[n]}{k}\to\R_{\geq 0}$ with a log-concave generating polynomial, and access to a distribution $p:[n]\to \R_{\geq 0}$ with an i.i.d.\ sampling oracle $\O$ such that for all $i$,
	\[ p(i)\geq \frac{\PrX{S\sim \mu}{i\in S}}{k+O(1)}, \]
	there is an algorithm that generates approximate samples from $\mu$ that are $\epsilon$-close in total variation distance, running in time $\poly(k,\log(1/\epsilon))$.
\end{theorem}
A natural question is, where does the oracle $\O$ come from? When $\mu$ is in approximately isotropic position, that is when $\PrX{S\sim \mu}{i\in S}$ is nearly the same for all $i$, then $\O$ can simply return a uniformly random element in $[n]$. See \cite{OR18} for some examples of isotropic distributions with log-concave generating polynomials (where the isotropy stems from symmetries in the ground set). 

Any non-isotropic distribution $\mu$ can be put in near-isotropic position by an operation we call subdivision, akin to linear transformation of continuous log-concave densities. In this operation a larger ground set $[m]$ together with a projection map $\pi:[m]\to [n]$ is used to define a new $\mu':\binom{[m]}{k}\to \R_{\geq 0}$: A sample $S\sim \mu'$ is obtained by first sampling $\set{e_1,\dots,e_k}\sim \mu$ and then replacing each $e_i$ with a \emph{uniformly random} element of $\pi^{-1}(e_i)$. The marginals of the new elements are of the form $\PrX{S\sim \mu}{i\in S}/\card{\pi^{-1}(i)}$. By choosing $\card{\pi^{-1}(i)}$ to be approximately proportional to $\PrX{S\sim \mu}{i\in S}$, near-isotropic position is achieved.

To keep the exposition clean, instead of ``transforming'' distributions $\mu$ into near-isotropic position, we instead change our algorithms to be aware of the marginals through the distribution $p$ and the oracle $\O$. This is similar in the continuous sampling literature where, instead of changing the distribution, we change, say the ball-walk algorithm, to an ``ellipsoid''-walk. Our algorithms perform exactly the same way as if we had performed a subdivision with infinitely large $m$, and we invite the reader to verify the equivalent subdivided forms of the algorithms.

For any given probability distribution $p:[n]\to\R_{\geq 0}$, an oracle $\O$ can be constructed with preprocessing time $O(n)$ that generates i.i.d.\ samples in time $O(\log n)$. So to leverage \cref{thm:main} it is enough to ``approximate'' the marginals of $\mu$, namely $\PrX{S\sim \mu}{i\in S}$ sufficiently well. Approximating the marginals can be done by approximate counting which reduces back to the approximate sampling task \cite{JVV86}. Naively one could use the ``down-up'' random walk to once-and-for-all approximate the marginals of $\mu$. Subsequently there is no need to recompute the marginals, and each subsequent sample can be generated in $\poly(k, \log n, \log(1/\epsilon))$ time. This approach is not satisfactory as the preprocessing step requires at least $\Omega(n)$ samples from $\mu$ to even cover all elements; with each sample taking time linear in $n$, the resulting running time will be quadratic in $n$. Instead we show how to use a careful cooling schedule, combined with \cref{thm:main}, to both approximate the marginals well enough, and to approximate the partition function.
\begin{theorem}\label{thm:oracle-construct}
	Given oracle access to a distribution $\mu$ with a log-concave generating polynomial, there is a randomized algorithm that constructs a sampling oracle $\O$ satisfying the assumptions of \cref{thm:main} in time $n\, \poly(k, \log n, \log(1/\delta))$ with probability $1-\delta$.
\end{theorem}
We show additionally that the partition function $\sum_{S} \mu(S)$ can be approximated using the same cooling schedule in time that improves significantly over $(n/\epsilon^2)\poly(k, \log n)$.
\begin{theorem}\label{thm:counting}
	Given oracle access to a density $\mu$ with a log-concave generating polynomial, there is a randomized algorithm that computes an $\epsilon$-relative error approximation of $\sum_{S} \mu(S)$ which runs in time $(n+1/\epsilon^2)\poly(k, \log n, \log(1/\delta))$ and succeeds with probability $1-\delta$.
\end{theorem}
As a corollary we obtain algorithms that can count bases of small-rank matroids with high precision in nearly linear time.
\begin{corollary}\label{cor:main}
	Given an oracle that answers independence queries for a matroid of rank $O(\poly\log n)$ over a ground set of $n$ elements, there is an algorithm that approximately counts bases within a multiplicative factor of $1+1/\sqrt{n}$ with high probability in nearly linear time $\tilde O(n)$.
\end{corollary}
By trivial reductions, this result also automatically allows us to estimate the number of small independent sets in nearly linear time, regardless of the rank of the matroid. For example, we can use this result to count forests of size $k=O(\poly\log n)$ in graphs in nearly linear time; for some motivating applications of counting forests, see \cite{GKRZ14}.
  
\subsection{Related Work}

Our results are related to recent work on accelerated sampling for determinantal point processes (DPP), which form a subset of distributions with log-concave polynomials and have a variety of applications in machine learning \cite{KT12}, statistics \cite{BLXA17}, and graph theory \cite{Gue83} (e.g., uniform sampling of spanning trees). DPPs enjoy a number of properties not satisfied by the general class of distributions with log-concave polynomials, such as closed form expressions for the partition function and the marginals, which make accelerated sampling easier. In particular, recent results by \cite{DWH18,DWH19,Der19,DCV19} take advantage of these properties to obtain $\poly(k)$ time sampling algorithms based on importance sampling proportional to the marginals. Crucially, these algorithms take advantage of the special structure of DPPs (such as the closed form of the partition function), which is why they cannot be directly extended to general distributions with log-concave polynomials. Even slight variants of DPPs, such as exponentiated DPPs \cite{MSJ18} do not enjoy the closed form expressions of DPPs and need our new framework. The differences in our approach, which make accelerated sampling possible for this broader class, include a hierarchical Markov chain procedure to avoid computing an exact partition function, and a cooling schedule for efficiently approximating the marginals. 

Negative dependence inequalities related to those we prove in our analysis have been studied for the class of strongly Rayleigh (SR) distributions \cite{BBL09}. This class includes all DPPs, but not all distributions with log-concave polynomials. For example, a uniform distribution over the bases of a matroid is not SR if the matroid is not balanced. Surprisingly, despite extensive literature on the negative correlation properties of SR distributions \cite{PP14}, the negative dependence inequality we prove for all distributions with log-concave polynomials appears to be new even for SR distributions. Some special forms of approximate negative correlation have been obtained for matroids and log-concave polynomials \cite{HSW18}; as an additional corollary of our proof techniques we rederive these correlation bounds and significantly generalize them.

The question of approximately sampling or counting bases of a matroid has been studied for a long time \cite[see, e.g.,][]{FM92, JS02, GJ18}, but most of the attention has been focused on proving just polynomial time efficiency, with some exceptions. Sampling random spanning trees can now be done in nearly linear time \cite{Sch18, ALOV20}, but for the seemingly related problem of sampling random forests \cite{GKRZ14}, the jury is still out. When it comes to counting algorithms, the situation is much worse since most results are only based on powerful but generic counting to sampling reductions \cite{JVV86}; these reductions often blow up the running time and make the algorithm impractical. This is despite the fact that the class of reliability problems for graphs, code, truss systems, etc. \cite{CC97, Cam98}, are all natural counting questions.

\subsection{Techniques}

Unlike most prior work on sampling from matroids, our algorithm does not just perform a walk on the basis exchange graph. Rather we combine walks on the basis exchange graph with macro steps that choose a small important subset of the ground set, and only permit walks on that small part. This significantly speeds up the sampling algorithm, by not allowing the basis exchange walks to focus on unimportant elements.

One of the key ingredients in our proof is a set of negative dependence inequalities, which yield interesting facts about matroids. If $S$ is a random basis of a matroid, think of a random forest of size $k$ in a graph, and $T$ is a fixed set, our inequalities upper bound probabilities of the following types of events: $S=T$, $S\supseteq T$, $S\subseteq T$, based on marginals of $S$. We apply recently derived Modified Log-Sobolev Ineqaulities for matroids and log-concave polynomials \cite{CGM19} in novel ways to derive these inequalities.

\subsection{Acknowledgements}

The first author thanks Jan Vondr\'{a}k for stimulating discussions related to negative dependence inequalities.
The second author thanks the NSF for funding via the NSF TRIPODS program.
	\section{Preliminaries}
\label{sec:prelims}

We use $[n]$ to denote the set $\set{1,\dots,n}$, and $\binom{[n]}{k}$ to denote the family of size $k$ subsets of $[n]$. All $\log$s are taken in base $e$.

We denote a sequence of length $t$ by $\sigma=\langle \sigma_1,\dots,\sigma_t\rangle$ and for an index set $I\subseteq [t]$ we let
\[ \sigma_I:=\set{\sigma_i\given i\in I}. \]

We use one of the common forms of the Chernoff bound.
\begin{lemma}[Chernoff Bound]\label{lem:chernoff}
Let $X_1,...,X_m$ be independent Bernoulli variables and let $S=\sum_iX_i$. Then, for any
$\epsilon\in(0,1)$, we have:
\begin{align*}
  \Pr*{\abs{S-\Ex{S}}\geq \epsilon\cdot\Ex{S}}\leq 2e^{-\epsilon^2\Ex{S}/3}.
\end{align*}
\end{lemma}

We denote the directional derivative operator in direction $v\in \R^n$ with $\partial_v$:
\[ \partial_v=v_1\partial_1+\dots+v_n\partial_n. \]

We use $\R[z_1,\dots,z_n]$ to denote the set of polynomials with real coefficients in variables $z_1,\dots,z_n$. We call a polynomial homogeneous (of degree $k$) if all of its terms have the same degree (equal to $k$). We call a polynomial multiaffine if no variable in it appears with degree more than $1$.

For a distribution or density function $\mu:\binom{[n]}{k}\to \R_{\geq 0}$, we define the generating polynomial $g_\mu$ to be
\[ g_\mu(z_1,\dots,z_n):=\sum_{S\in \binom{[n]}{k}} \mu(S)\prod_{i\in S} z_i. \]
Note that by definition $g_\mu$ is both multiaffine and homogeneous.

We use $\D{\nu \mid \mu}$ to denote the Kullback-Leibler divergence between distributions $\nu$ and $\mu$ defined as follows:
\[ \D{\nu \mid \mu}:=\ExX*{S\sim \mu}{\frac{\nu(S)}{\mu(S)}\log \frac{\nu(S)}{\mu(S)}}=\ExX*{S \sim \nu}{\log\frac{\nu(S)}{\mu(S)}}.  \]
We use $\norm{\nu-\mu}_\tv$ to denote the total variation distance between distributions $\nu$ and $\mu$:
\[ \norm{\nu-\mu}_\tv := \frac{1}{2}\sum_{S} \abs{\nu(S)-\mu(S)}. \]

\subsection{Markov Chains}

For a Markov chain $P$ with stationary distribution $\mu$, we define the mixing time $\tmix(P, \epsilon)$ to be the minimum time $t$ such that for all starting states $S\in \supp(\mu)$
\[ \norm{P^t(S, \cdot)-\mu}_\tv \leq \epsilon. \]

\begin{theorem}\label{thm:mcmc-classical}
If an irreducible aperiodic Markov chain with stationary distribution $\mu$ and transition matrix $P$
  satisfies $\norm{P^t(S,\cdot)-\mu}_\tv\leq1/4$ for all $S\in \supp(\mu)$
  and some $t\geq 1$, then for any $\epsilon\in (0, 1/4]$
  \[ \tmix(P, \epsilon) \leq t \log(1/\epsilon).\]
\end{theorem}


\subsection{Log-Concave Polynomials}

We call a polynomial $g(z_1,\dots,z_n)$ with nonnegative coefficients log-concave if it is log-concave as a function over $\R_{\geq 0}^n$, i.e., for any $x, y\in \R_{\geq 0}^n$ and $\lambda\in (0, 1)$ we have
\[ g(\lambda x+(1-\lambda)y)\geq g(x)^\lambda g(y)^{1-\lambda}. \]

One of the key operations preserving log-concavity of a polynomial is composition with a linear operator $T:\R^m\to \R^n$ for which $T(\R_{\geq 0}^m)\subseteq \R_{\geq 0}^n$ \cite{AOV18}.

Log-concavity of a polynomial is in general not preserved under differentiation. Prior work has considered two classes of polynomials, called strongly log-concave \cite{Gur09}, and completely log-concave or Lorentzian \cite{AOV18, BH19}, to deal with this issue. However, not being closed under derivatives is an artifact of high or low powers of variables dominating others and masking non-log-concavity at middle scales. In particular, for homogeneous multiaffine polynomials (which includes all of the polynomials considered here), this masking does not happen and all three notions coincide.
\begin{lemma}[{\cite[see, e.g.,][]{ALOV20}}]\label{lem:clc}
	Let $g\in \R[z_1,\dots,z_n]$ be a multiaffine homogeneous polynomial with nonnegative coefficients. If $g$ is log-concave, then it is completely log-concave as well, which means that for any $k\in \Z_{\geq 0}$ and directions $v_1,\dots, v_k\in \R_{\geq 0}^n$, the following polynomial is log-concave:
	\[ \partial_{v_1}\cdots \partial_{v_k} g. \]
\end{lemma}


\subsection{Down and Up Operators}

We define two operators borrowed from the literature on high-dimensional expanders \cite{KM16, DK17, KO20}. The down operator acts on the distribution of a random set $S$ of size $k$ and extracts a uniformly random subset of size $l$.
\begin{definition}\label{def:down}
	For nonnegative integers $l\leq k\leq n$, let $D_{k\to l}\in \R^{\binom{[n]}{k}\times \binom{[n]}{l}}$ be the $k$ to $l$ down operator defined as
	\[ D_{k\to l}(S, T)=\begin{cases}
		\frac{1}{\binom{k}{l}}& \text{if }T\subseteq S,\\
		0 & \text{otherwise}.\\
	\end{cases}
 	\]
\end{definition}
Note that if $\mu:\binom{[n]}{k}\to \R_{\geq 0}$ is a distribution on sets of size $k$, then $\mu D_{k\to l}$ is a distribution on sets of size $l$. Further note that the down operators compose in the expected way:
\[ D_{k\to l}D_{l\to m}=D_{k\to m}. \]

When there is a background distribution $\mu:\binom{[n]}{k}\to \R_{\geq 0}$ we can define an up operator, which is the time-reversal of the down operator for stationary distribution $\mu$.
\begin{definition}
	For nonnegative integers $l\leq k\leq n$ and background distribution $\mu:\binom{[n]}{k}\to \R_{\geq 0}$, let $U_{l\to k}\in \R^{\binom{[n]}{l}\times \binom{[n]}{k}}$ be the up operator defined as
	\[ 
		U_{l\to k}(T, S)=\begin{cases}
			\frac{\mu(S)}{\sum_{S'\supseteq T} \mu(S')}& \text{if }T\subseteq S,\\
			0& \text{otherwise}.
		\end{cases}
	\]
\end{definition}

It is easy to see that $D_{k\to l}U_{l\to k}$ defines a time-reversible Markov chain on the state space $\binom{[n]}{k}$ with stationary distribution $\mu$. This can be seen by considering a bipartite graph between $\binom{[n]}{k}$ and $\binom{[n]}{l}$ with an edge of weight $\mu(S)$ between $S$ and $T$ whenever $T\subseteq S$. Then $D_{k\to l}$ is the operation of randomly walking from the top side to the bottom side, and $U_{l\to k}$ is randomly walking from the bottom to the top. Further the weighted degree of every node on the top is $\binom{k}{l}\mu(S) \propto \mu(S)$, which makes $\mu(S)$ the stationary distribution of $D_{k\to l}U_{l\to k}$.

These Markov chains are efficiently implementable, with oracle access to $\mu$, when $k-l=O(1)$. The difficult part is implementing $U_{l\to k}$ which naively takes time $n^{k-l}$. A sequence of works \cite{ALOV19, CGM19, ALOV20} have identified tight mixing times for distributions $\mu$ with a log-concave generating polynomial, when $l=k-1$:
\begin{theorem}[\cite{ALOV20}]\label{thm:down-up}
	Let $\mu:\binom{[n]}{k}\to \R_{\geq 0}$ be a distribution with a log-concave generating polynomial. Then $P:=D_{k\to k-1}U_{k-1\to k}$ is a Markov chain with stationary distribution $\mu$ and
	\[ \tmix(P, \epsilon)\leq O(k\log(k/\epsilon)). \]
\end{theorem}

One of the important steps yielding the above result is a discrete-time variant of the Modified Log-Sobolev Inequality (MLSI) proved by \textcite{CGM19}. We will use the main ingredient from the proof of the MLSI inequality, that the down operator shrinks KL-divergence:
\begin{theorem}[\cite{CGM19}]\label{thm:MLSI}
	If $\mu:\binom{[n]}{k}\to \R_{\geq 0}$ is a distribution with a log-concave generating polynomial, and $\nu:\binom{[n]}{k}\to \R_{\geq 0}$ is an arbitrary distribution, then
	\[ \D{\nu D_{k\to k-1}\mid \mu D_{k\to k-1}}\leq \frac{k-1}{k}\D{\nu \mid \mu}. \]
\end{theorem}
	\section{Instant Mixing via Isotropy}
\label{sec:instant}

Let $\mu:\binom{[n]}{k}\to \R_{\geq 0}$ be a
density function on sets of size $k$. Here we describe a general random-walk-based algorithm for approximately sampling a set $S$ with $\Pr{S}\propto \mu(S)$, and show that isotropy of $\mu$ guarantees almost-instantaneous mixing. However, implementing each step of this algorithm is nontrivial, and later in \cref{sec:hierarchical} we show how to implement the steps using an \emph{inner random walk}, proving \cref{thm:main}.

Without loss of generality, assume that
every element of $[n]$ is drawn by $\mu$ with positive probability,
i.e., that $\PrX{\Sc\sim \mu}{i\in \Sc}>0$ for all $i\in[n]$. Moreover,
let $p:[n]\to \R_{\geq0}$ be a probability distribution over
$[n]$ such that $p(i)>0$ for all $i$. We assume there is an oracle $\O$ that can produce i.i.d.\ samples from $p$. Consider the following Markov
chain $\Mc_{\mu,p}^t$ defined for any positive integer $t$, with the state space
$\supp(\mu)$. Starting from $S_0\in\supp(\mu)$, one
step of the chain is given by:
\begin{enumerate}
\item Draw a sequence
  $\rho=\langle \rho_1,...,\rho_t\rangle$ of i.i.d.~samples from $p$ by calling $\O$.
\item Arrange $S_0$ together with $\rho$, permuting the elements uniformly at random to
  obtain a sequence $\sigma=\langle \sigma_1,...,\sigma_{t+k}\rangle$.
\item Return $\Sc_1=\sigma_{\Sc}$, where $\Sc$ is drawn from a
distribution
  $\mu_{\sigma,p}:\binom{[t+k]}{k}\to \R_{\geq0}$ defined by:
  \begin{align*}
    \mu_{\sigma,p}(S) = \frac{\mu(\sigma_S)\prod_{i\in
    S}p(\sigma_i)^{-1}}{Z_{\mu,p}(\sigma)}\quad\text{where}\quad
    \sigma_S=\set{\sigma_i\given i\in S}.
  \end{align*}
\end{enumerate}
Here, $Z_{\mu,p}(\sigma)$ denotes the normalization constant of
$\mu_{\sigma,p}$. We first establish some basic properties of 
$\Mc_{\mu,p}^t$. In particular, we show that, regardless of the choice of the
sampling distribution $p$, the stationary distribution of $\Mc_{\mu,p}^t$ is $\mu$.
\begin{lemma}\label{lem:convergence}
If $\mu$ has a log-concave generating polynomial then for any $t\geq
1$ the chain $\Mc_{\mu,p}^t$ is irreducible and aperiodic with
  stationary distribution $\mu$.
\end{lemma}
\begin{proof}
Let $P$ denote the transition probability matrix of
$\Mc_{\mu,p}^t$. The aperiodicity follows because for any
$S\in\supp(\mu)$, we have $P(S,S)>0$. To establish irreducibility,
note that since $t\geq 1$, for any $S,S'\in\supp(\mu)$ that differ
only by swapping a pair of elements (i.e., $|S\cap S'|=k-1$), we have
$P(S,S')>0$. Since $\mu$ has a log-concave generating polynomial,
$\supp(\mu)$ is a matroid so there is a sequence of such
swaps within $\supp(\mu)$ that converts any $S\in\supp(\mu)$ into any $S'\in\supp(\mu)$.

It remains to find the stationary
distribution. To that end, suppose that we perform one step of the
chain starting from $\Sc_0\sim\mu$. We first derive the distribution
of the intermediate sequence $\sigma$, using $r=t+k$ and
$p(\tau)=\prod_{i=1}^rp(\tau_i)$ for $\tau\in[n]^r$ as shorthands:
\begin{align}
  \Prob[\sigma = \tau] = \frac1{k!{r\choose k}}\sum_{S\in{[r]\choose
  k}}\mu(\tau_S)\prod_{i\in [r]\backslash S}p(\tau_i) =
  \frac{p(\tau)}{k!{r\choose k}}\sum_{S\in{[r]\choose
  k}}\mu(\tau_S)\prod_{i\in S}p(\tau_i)^{-1} =
  \frac{p(\tau) Z_{\mu,p}(\tau)}{k!{r\choose k}},\label{eq:intermediate}
\end{align}
where the summation enumerates different placements of $\Sc_0$ in the
sequence $\sigma$ resulting from  uniformly permuting the
sequence. Finally, we derive the distribution of $\Sc_1$ as follows:
\begin{align*}
  \Prob[\Sc_1=S_1]
  &=
    \sum_{\tau\in[n]^r}\Prob[\sigma_{\Sc}=S_1\mid\sigma=\tau]\Prob[\sigma=\tau]
  =\sum_{\tau\in[n]^r}\sum_{S\in{[r]\choose k}}
    \1_{[\tau_S=S_1]}\mu_{\tau,p}(S)\Prob[\sigma=\tau]\\
  &=\sum_{S\in{[r]\choose
    k}}\sum_{\tau\in[n]^r}\1_{[\tau_S=S_1]}
    \frac{\mu(\tau_S)\prod_{i\in
    S}p(\sigma_i)^{-1}}{Z_{\mu,p}(\tau)}\,
    \frac{Z_{\mu,p}(\tau)p(\tau)}{k!{r\choose k}}\\
  &=k!{r\choose
    k}\sum_{\tilde\tau\in[n]^t}\mu(S_1)\frac{p(\tilde\tau)}{k!{r\choose k}}
    \ =\ \mu(S_1)\!\!\sum_{\tilde\tau\in[n]^t}p(\tilde\tau)\ =\ \mu(S_1),
\end{align*}
which means that $\mu$ is the stationary distribution of
$\Mc_{\mu,p}^t$, concluding the proof.
\end{proof}

Since $\Mc_{\mu,p}^t$ is irreducible and aperiodic, it converges to
its stationary distribution from any starting state. We next establish
the rate of this convergence, under some additional assumptions. 

\subsection{Negative Dependence Property}
To further analyze the Markov chain $\Mc_{\mu,p}^t$, we establish a new
property of distributions with a
log-concave generating polynomial. This property, which may be of
independent interest, allows us to bound the normalization 
constant $Z_{\mu,p}(\sigma)$, as shown in the
following lemma (the proof is deferred to \cref{sec:negdep}).
\begin{lemma}\label{lem:condition}
Let $\mu:\binom{[n]}{k}\to \R_{\geq 0}$ be a distribution with a
log-concave generating polynomial. Then, for any function $p:[n]\rightarrow \R_{>0}$,
integer $t\geq k$ and  sequence $\tau\in[n]^t$, we have
\begin{align*}
\sum_{S\subseteq {[t]\choose k}}\mu(\tau_S)\prod_{i\in
S}p(\tau_i)^{-1}
  \leq \bigg(\sum_{i=1}^t
  \frac{\Prob_{\Sc\sim\mu}[\tau_i\in\Sc]}{k\cdot p(\tau_i)}\bigg)^k.
\end{align*}
\end{lemma}
 Note that when $p$ is the sampling distribution from
$\Mc_{\mu,p}^t$, then the left-hand side of the inequality is 
exactly the normalization constant $Z_{\mu,p}(\tau)$. To give some
intuition behind this bound, suppose that $p(i)=1$ for all
$i\in[n]$, and let $\tau$ be a sequence of $t$ unique elements from
some set $T$. Then, the left-hand side can be concisely written as
$\PrX{\Sc\sim \mu}{\Sc\subseteq T}$. Furthermore, let $s_1,...,s_k$ be a random permutation of
the elements of $\Sc$ (i.e., $s_i$ are identically distributed
according to the marginal distribution of $\mu$, but
they may not be independent). Then, the inequality can be 
stated as a new negative correlation property that may be of independent
interest: 
\begin{align}
\Prob\bigg[\bigwedge_{i=1}^k \,s_i\in T \bigg]\leq
\bigg(\frac1k\sum_{i\in
  T}\Prob[i\in\Sc]\bigg)^k\!=  \prod_{i=1}^k\Prob[s_i\in T], \quad \text{for all $T$ of size $\geq k$}.\label{eq:neg-dep}
\end{align}
When the size of $T$ is $k$, then \cref{eq:neg-dep} is implied by an even
stronger property of log-concave 
distributions:
\begin{align*}
\Prob\bigg[\bigwedge_{i=1}^k \,s_i\in T \bigg]=  \Prob[\Sc=T]\overset{(a)}{\leq} \prod_{i\in
  T}\Prob[i\in\Sc]\overset{(b)}{\leq} \bigg(\frac1k\sum_{i\in
  T}\Prob[i\in\Sc]\bigg)^k,
\end{align*}
where $(a)$ will be proved in \cref{sec:negdep} and $(b)$ is the
arithmetic-geometric mean inequality. For $t>k$,
\cref{eq:neg-dep} no longer follows from $(a)$. For
example, suppose that $\Prob[i\in S]=\frac kn$ for all $i$. Then, for a
set $T$ of size $t>k$, inequality $(a)$ implies that
$\Prob[\Sc\subseteq T]\leq {t\choose k} (\frac kn)^k$, whereas
\cref{eq:neg-dep} states that $\Prob[\Sc\subseteq T]\leq (\frac tn)^k$,
which is tighter by a factor that can be as large as $2^k$.

\subsection{Coupling Argument}
To establish the convergence rate of $\Mc_{\mu,p}^t$, we couple a single
step of this Markov chain with a single step of the chain
$\Mc_{\mu,p}^{t-k}$. Crucially, we allow the former to start from a
fixed arbitrary state $S_0\in\supp(\mu)$, but we assume that the latter is already
at the stationary distribution, with starting state
$\Sc_0'\sim\mu$. The overall coupling procedure is illustrated
in \cref{fig:coupling}.

We next bound the probability that the output
states of the two chains are different, i.e., $\Prob[\Sc_1\neq \Sc_1']$,
which implies a bound on the total variation distance of $\Mc_{\mu,p}^t$ from
$\mu$ after one step. To achieve this, in addition to $\mu$ having a
log-concave polynomial, we must also assume that the importance
sampling distribution $p$ is a sufficiently good approximation of the
marginal distribution of $\mu$, i.e., that $p(i)\approx
\frac1k\Prob_{\Sc\sim\mu}[i\in \Sc]$. This ensures the tightest form of the
inequality in \cref{lem:condition}. It also makes intuitive sense,
since if an element $i$ has high marginal probability, then we need to
ensure that the chain is likely to find it in its i.i.d.~sampling
phase. The precise formulation of the result is given in the following lemma.
\begin{lemma}\label{lem:coupling}
Let $\mu$ be a $k$-homogeneous distribution with a
log-concave polynomial. For any $c>1$, if $t\geq ck^2$ and
$k\cdot p(i)\geq (1+\frac1{ck})^{-1}\Prob_{S\sim\mu}[i\in S]$ for all $i\in[n]$, then the
transition matrix $P$ of $\Mc_{\mu,p}^t$ satisfies for any starting state $S_0$:
\begin{align*}
  \norm*{P(S_0,\cdot) - \mu}_\tv\leq 5/c.
\end{align*}
\end{lemma}
\begin{figure}
\centering
  \begin{tikzpicture}
    \draw (0,0) node (A1) {$S_0$};                              
    \draw (2,0) node (A2){$\Mc_{\mu,p}^t$};                              
    \draw (5,0) node (B1)  {$\Sc_0'$};                              
    \draw (7,0) node (B2) {$\Mc_{\mu,p}^{t-k}$};
    \draw (2,-1) node (A3) {$\langle \rho_1,...,\rho_t\rangle$};
    \draw (7,-1) node (B3) {$\langle \rho_1',...,\rho_{t-k}'\rangle$};
    \draw (2,-2) node (A4){$\langle \sigma_1,...,\sigma_t,...,,\sigma_{t+k}\rangle$};
    \draw (7,-2) node (B4) {$\langle
      \sigma_1',...,\sigma_{t-k}',...,\sigma_t'\rangle$};
    \draw (2,-3) node (A5) {$\Sc_1$};                              
    \draw (7,-3) node (B5) {$\Sc_1'$};
    \draw [->] (A1) -- (A2);
    \draw [->] (B1) -- (B2);
    \draw [->] (A2) -- (A3);
    \draw [->] (B2) -- (B3);
    \draw [->] (A3) -- (A4);
    \draw [->] (B3) -- (B4);
    \draw [->] (A4) -- (A5);
    \draw [->] (B4) -- (B5);
    \path[very thick,<->]  (A3.east) edge node[sloped, anchor=center,
    above] {\footnotesize coupled} (B4.west);
    \path[very thick,<->]  (A5) edge node[sloped, anchor=center,
    above] {\footnotesize coupled} (B5);
  \end{tikzpicture}
  \caption{The coupling between $\Mc_{\mu,p}^t$
    started at fixed $S_0\in\supp(\mu)$ and $\Mc_{\mu,p}^{t-k}$
    started at $\Sc_0'\sim\mu$.}\label{fig:coupling}
\end{figure}
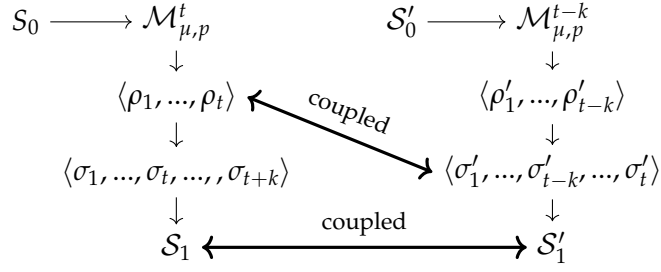
\begin{proof}
As discussed above, it suffices to show that the chain $\Mc_{\mu,p}^t$
started at $S_0$ can be coupled with chain $\Mc_{\mu,p}^{t-k}$ started
at $\Sc_0'\sim\mu$, so that $\Prob[\Sc_1=\Sc_1']\geq 1-\frac 5c$. First,
letting $\epsilon=\frac1{ck}$, observe that by combining
\cref{lem:condition} with the approximation guarantee $k\cdot
p(i)\geq (1+\epsilon)^{-1}\Prob_{\Sc\sim\mu}[i\in \Sc]$, we can upper bound
the normalization constant of the distribution $\mu_{\tau,p}$ for
any $\tau\in[n]^t$ as follows:
\begin{align*}
  Z_{\mu,p}(\tau) \leq
  \bigg(\sum_{i=1}^t\frac{\Prob_{\Sc\sim\mu}[\tau_i\in \Sc]}{k\cdot
  p(\tau_i)}\bigg)^k\leq (t+\epsilon t)^k.
\end{align*}
We use this bound first to couple the intermediate sequences generated
by the two chains, namely $\rho=\langle\rho_1,...,\rho_t\rangle$ and
$\sigma'=\langle\sigma_1',...,\sigma_t'\rangle$ (see \cref{fig:coupling}). Note that the lengths
of the sequences match because the former does not have the starting
state $S_0$ permuted into it. Using the derivation in
\cref{eq:intermediate}, we bound the probability distribution of
$\sigma'$ in terms of the distribution of $\rho$:
\begin{align*}
  \Prob[\sigma'\!=\tau] = \frac{p(\tau)Z_{\mu,p}(\tau)}{k!{t\choose
  k}}\leq p(\tau)\frac{(t+\epsilon t)^k}{k!{t\choose
  k}}=\Prob[\rho=\tau]\prod_{i=0}^{k-1}\frac{t+\epsilon t}{t-i}\leq
  \Prob[\rho=\tau]\bigg(\frac{t+\epsilon t}{t-k}\bigg)^k,
\end{align*}
so using the assumptions that $t\geq ck^2$ and $\epsilon=\frac1{ck}$
we conclude that $\Prob[\rho\!=\!\tau]\geq
(1-\frac 2c)\Prob[\sigma'\!\!=\!\tau]$ for all $\tau\in[n]^t$, which means that by coupling we
can ensure that $\Prob[\rho=\sigma']\geq 1-\frac2c$.

We next couple the output sets $\Sc_1$ and $\Sc_1'$ by first
conditioning on the event that $\rho=\sigma'$, and then observing that
distributions $\Sc\sim\mu_{\sigma,p}$ and $\Sc'\sim\mu_{\sigma',p}$ are identical as
long as $\Sc$ does not select any indices corresponding to the input
set $S_0$. Without loss of generality, fix the permutation of $\sigma$
so that $\sigma=\langle\rho_1,...,\rho_t,s_1,...,s_k\rangle$, where
$S_0=\{s_1,...,s_k\}$. Letting $r=t+k$, we now lower bound the probability 
that $\Sc$ does not select any of the last $k$ indices:
\begin{align*}
  \Prob\big[\Sc\subseteq[t]\big]=\sum_{\tau\in[n]^t}\Prob[\rho=\tau]\sum_{S\subseteq[t]}\frac{\mu(\tau_S)\prod_{i\in
  S}p(\tau_i)^{-1}}{Z_{\mu,p}(\langle\tau,s_1,...,s_k\rangle)}\geq
  \sum_{\tau\in[n]^t}\frac{p(\tau)Z_{\mu,p}(\tau)}{(r+\epsilon
  r)^k}=\frac{k!{t\choose k}}{(r+\epsilon r)^k}\geq 1-\frac3c.
\end{align*}
Observe that, conditioned on the events that $\rho=\sigma'$ and that
$\Sc\subseteq[t]$, the distributions of $\Sc_1$ and $\Sc_1'$ are
identical. A union bound shows that the two events occur together with 
probability at least $1- \frac5c$, completing the proof.
\end{proof}
Setting $c=20$ in \cref{lem:coupling}, we obtain that if $t\geq
20k^2$ and $k\cdot p(i)\geq
\big(1+\frac1{20k}\big)^{-1}\Prob_{\Sc\sim\mu}[i\in\Sc]$, then the
transition matrix $P$ of $\Mc_{\mu,p}^t$ for any state
$S_0\in\supp(\mu)$ satisfies $\|P(S_0,\cdot)-\mu\|_\tv\leq 1/4$. This,
combined with \cref{lem:convergence} and
\cref{thm:mcmc-classical}, implies that:
\begin{align*}
  \sup_{S_0\in\supp(\mu)}\big\|P^s(S_0,\cdot)-\mu\big\|_\tv\leq\epsilon\quad\text{for}\quad
  s\geq \log1/\epsilon.
\end{align*}

	\section{Negative Dependence Inequalities}
\label{sec:negdep}

In this section we prove \cref{lem:condition}. As a corollary of our techniques, we will additionally prove an approximate negative correlation inequality which generalizes results of \textcite{HSW18} and may be of independent interest. First we strengthen \cref{thm:MLSI} to show how KL-divergences contract under the $D_{k\to l}$ operator.
\begin{lemma}\label{lem:kl-shrink}
	If $\mu:\binom{[n]}{k}\to \R_{\geq 0}$ is a distribution with a log-concave generating polynomial, and $\nu:\binom{[n]}{k}\to \R_{\geq 0}$ is an arbitrary distribution, then
	\[ \D{\nu D_{k\to l}\mid \mu D_{k\to l}}\leq \frac{l}{k}\D{\nu\mid \mu}. \]
\end{lemma}
\begin{proof}
	We prove this by induction on $k-l$. When $k-l=1$, this is already the contents of \cref{thm:MLSI}. Note that the distribution $\mu D_{k\to l}$ always has a log-concave generating polynomial. This is because
	\[ g_{\mu D_{k\to l}}\propto (\partial_1+\dots+\partial_n)^{k-l} g_\mu  \]
	and by \cref{lem:clc}, the above polynomial is log-concave. So assuming the statement is correct for $k$ and $l$, we can prove it for $k$ and $l-1$ by applying \cref{thm:MLSI} to the distributions $\nu D_{k\to l}$ and $\mu D_{k\to l}$; we obtain
	\[ \D{\nu D_{k\to l}D_{l\to l-1} \mid \mu D_{k\to l}D_{l\to l-1} }\leq \frac{l-1}{l}\D{\nu D_{k\to l}\mid \mu D_{k\to l}}. \]
	But note that $D_{k\to l}D_{l\to l-1}=D_{k\to l-1}$. Combining this with the shrinkage of KL-divergence for $D_{k\to l}$ we get
	\[ \D{\nu D_{k\to l-1}\mid \mu D_{k\to l-1}}\leq \frac{l-1}{l}\cdot \frac{l}{k} \D{\nu\mid \mu}=\frac{l-1}{k}\D{\nu\mid \mu}, \]
        which completes the proof.
\end{proof}

In the rest of this section we will apply \cref{lem:kl-shrink} for $l=1$. Note that $\mu D_{k\to 1}$ is a distribution on $\binom{[n]}{1}\cong [n]$, which assigns a weight of $\frac{1}{k}\PrX{S\sim \mu}{i\in S}$ to every singleton $\set{i}$. As a warmup, we first prove a form of negative correlation for distributions with a log-concave generating polynomial.
\begin{proposition}\label{prop:neg1}
	If $\mu:\binom{[n]}{k}\to \R_{\geq 0}$ is a distribution with a log-concave generating polynomial, then for any set $T\in \binom{[n]}{k}$,
	\[ \mu(T)=\PrX{S\sim \mu}{S=T}\leq \prod_{i\in T} \PrX{S\sim \mu}{i\in S}.  \]
\end{proposition}
\begin{proof}
	Let $\nu:\binom{[n]}{k}\to \R_{\geq 0}$ be defined to be $0$ everywhere, except at $T$, where $\nu(T)=1$. Then
	\[ \D{\nu \mid \mu}:= \ExX*{S\sim \nu}{\log \frac{\nu(S)}{\mu(S)}}=\log \frac{1}{\mu(T)}. \]
	Note that $\nu D_{k\to 1}$ is the uniform distribution over elements of $T$. Therefore
	\[ \D{\nu D_{k\to 1}\mid \mu D_{k\to 1}}=\ExX*{\set{i}\sim \nu D_{k\to 1}}{\log \frac{\nu D_{k\to 1}(\set{i})}{\mu D_{k\to 1}(\set{i})}} =\frac{1}{k}\sum_{i\in T}\log \frac{1/k}{\PrX{S\sim \mu}{i\in S}/k}. \]
	Using \cref{lem:kl-shrink} for $l=1$, we obtain that
	\[ \sum_{i\in T} \log \frac{1}{\PrX{S\sim \mu}{i\in S}}\leq \log \frac{1}{\mu(T)}. \]
	Exponentiating and rearranging the above yields the desired inequality.
\end{proof}

Next we prove a negative correlation property that generalizes some results of \textcite{HSW18}. This result can be interpreted as negative correlation up to a factor of $2$, and it may be of independent interest.
\begin{lemma}
	If $\mu:\binom{[n]}{k}\to\R_{\geq 0}$ is a distribution with a log-concave generating polynomial, then for any set $T$ of size $l\leq k$, we have
	\[ \PrX{S\sim \mu}{T\subseteq S}\leq \alpha(k, l)\cdot \prod_{i\in T}\PrX{S\sim \mu}{i\in S}, \] 
	where $\alpha(k, l)=\binom{k}{l}(\frac{l}{k})^l\leq l^l/l!$. In particular for $l=2$, $\alpha(k, l)\leq 2$, and more generally $\alpha(k, l)\leq e^l$.
\end{lemma}
\begin{proof}
	We apply \cref{prop:neg1} to the distribution $\mu D_{k\to l}$. Note that the marginals of $\mu D_{k\to l}$ are simply $l/k$ times those of $\mu$, because in $D_{k\to l}$ we only keep $l$ out of $k$ elements uniformly at random:
	\[ \PrX{S\sim \mu D_{k\to l}}{i\in S}=\frac{l}{k}\PrX{S\sim \mu}{i\in S}. \]
	Further $\mu D_{k\to l}(T)$ is simply the chance that a sample $S\sim \mu$ contains $T$ times the chance that in dropping from $k$ to $l$ elements, we exactly choose elements of $T$, which is $1/\binom{k}{l}$:
	\[ \mu D_{k\to l}(T)=\frac{1}{\binom{k}{l}}\PrX{S\sim \mu}{T\subseteq T}. \]
	Putting these together we get
	\[ \frac{1}{\binom{k}{l}}\PrX{S\sim \mu}{T\subseteq T} \leq \prod_{i\in T}\parens*{\frac{l}{k}\PrX{S\sim \mu}{i\in S}}. \]
\end{proof}
 
So far we have been bounding the probability of $\PrX{S\sim \mu}{T\subseteq S}$. Next, we prove a bound on the event that $S\subseteq T$. This is the key ingredient for proving \cref{lem:condition}.

\begin{lemma}\label{lem:neg2}
	If $\mu:\binom{[n]}{k}\to\R_{\geq 0}$ is a distribution with a log-concave generating polynomial, then for any set $T$, we have
	\[ \PrX{S\sim \mu}{S\subseteq T}\leq \parens*{\frac{\sum_{i\in T} \PrX{S\sim \mu}{i\in S}}{k}}^k. \]
\end{lemma}
\begin{proof}
	Consider the distribution $\nu:\binom{n}{k}\to \R_{\geq 0}$ defined as $\mu$ conditioned on the sampled set being $\subseteq T$:
	\[ \nu(S)\propto \begin{cases}
			\mu(S)& S\subseteq T,\\
			0& \text{otherwise}.
		\end{cases} \]
		Notice that the normalizing constant in the above definition is $\PrX{S\sim \mu}{S\subseteq T}$. It is easy to see that
		\[ \D{\nu \mid \mu}=\log \parens*{\frac{1}{\PrX{S\sim\mu}{S\subseteq T}}}. \]
		Let $D_{k\to 1}\in \R^{\binom{n}{k}\times n}$ be the operator that maps a distribution on sets of size $k$ to a distribution on singletons, by first sampling a set from the original distribution, and then sampling a uniformly random element inside of it.
		\[ D_{k\to 1}(S, i)=\begin{cases}
			\frac{1}{k}& i\in S,\\
			0 & i\notin S.
	\end{cases}
	 \]
	Applying $D_{k\to 1}$ shrinks the KL-divergence by a factor of $k$ by \cref{lem:kl-shrink}. In other words
	\[ \D{\nu D_{k\to 1}\mid \mu D_{k\to 1}}\leq \frac{1}{k} \D{\nu \mid \mu}. \]
	It is easy to see that $\mu D_{k\to 1}(i)=\PrX{S\sim \mu}{i\in S}/k$. On the other hand, $\nu D_{k\to 1}$ is not easily expressible in terms of the marginals. But notice that $\nu D_{k\to 1}$ is a distribution on $[n]$ whose support is inside of $T$. We can ask, among all distributions supported on $T$, which one minimizes the KL-divergence w.r.t.\ $\mu D_{k\to 1}$? It is exactly the one obtained from $\mu D_{k\to 1}$ by conditioning on being inside $T$ \cite[see, e.g.,][]{CT12}. In other words, if we define
	\[ \omega(i)\propto \begin{cases}
		\mu D_{k\to 1}(i)& i\in T\\
		0& i\notin T,\\
		\end{cases}
	\]
	then
	 \[\D{\omega\mid \mu D_{k\to 1}}\leq \D{\nu D_{k\to 1} \mid \mu D_{k\to 1}}\leq \frac{1}{k}\D{\nu \mid \mu}. \]
	 With this definition of $\omega$, it is easy to compute $\D{\omega\mid \mu D_{k\to 1}}$:
	 \[ \D{\omega \mid \mu D_{k\to 1}}=\log\parens*{\frac{1}{\sum_{i\in T}\mu D_{k\to 1}(i)}}=\log\parens*{\frac{k}{\sum_{i\in T} \PrX{S\sim \mu}{i\in S}}}. \]
	 Putting this altogether we get that
	 \[ \log\parens*{\frac{k}{\sum_{i\in T} \PrX{S\sim \mu}{i\in S}}}\leq \frac{1}{k}\log\parens*{\frac{1}{\PrX{S\sim\mu}{S\subseteq T}}}.\]
	 Exponentiating and rearranging gives us the result.
\end{proof}

We are now ready to prove \cref{lem:condition}.
\begin{proof}[Proof of \cref{lem:condition}]
	First, note that if there are any repetitions in $\tau$, we can reduce the size of $\tau$ by merging all duplicates of an element $e$ together, and replacing $p(e)$ by $p(e)/m$, where $m$ is the number of duplicates. This does not change the r.h.s.\ of the desired inequality; previously there were $m$ terms of the form $\frac{\PrX{S\sim \mu}{e\in S}}{kp(e)}$, and now there is one term equal to $\frac{\PrX{S\sim \mu}{e\in S}}{kp(e)/m}$. Similarly, the l.h.s.\ does not change because every set $S$ containing $e$ previously appeared $m$ times, once for each copy of $e$, and now it appears once, but with $m$ times the value.
	
	Now that there are no duplicates in $\tau$, we can think of $\tau$ as a set $T$. The inequality we need to prove is
	\[ \sum_{S\in \binom{T}{k}} \parens*{\frac{\mu(S)}{\prod_{i\in S}p(i)}} \leq \parens*{\sum_{i\in T} \frac{\PrX{S\sim \mu}{i\in S}}{kp(i)}}^k. \]
	Note that if all $p(i)$ were equal to $1$, this would be exactly \cref{lem:neg2}.
	
	Our strategy is to first prove this inequality when $p$ takes integral values, by reducing to $p=1$, and then generalize to rational $p$ values. Note that both sides of the inequality are $k$-homogeneous in the $p$ values. So by scaling the inequality for integral $p$ implies the same inequality for rational $p$. Finally by continuity, we get the inequality for general $p:[n]\to \R_{\geq 0}$. So from now on assume that $p:[n]\to \Z_{\geq 0}$ takes only integer values.
	
	We will apply \cref{lem:neg2} to a subdivision of $\mu$. Consider a set of $m$ elements, where $m=p(1)+\dots+p(n)$. Let $\pi:[m]\to [n]$ be a projection defined by mapping $p(i)$ many distinct elements in $[m]$ to $i$. Now consider the distribution $\mu':\binom{[m]}{k}\to\R_{\geq 0}$ defined as follows: First we sample $\set{e_1,\dots, e_k}\sim \mu$, and then we replace each $e_i$ with a uniformly random element of $\pi^{-1}(e_i)$ to obtain a set $S'\in \binom{[m]}{k}$. This operation preserves log-concavity of the generating polynomial, because $g_{\mu'}$ can be obtained from $g_\mu$ by composing with a linear map:
	\[ g_{\mu'}(z_1,\dots,z_m)=g_\mu\parens*{\frac{\sum_{i\in \pi^{-1}(1)}z_i}{p(1)},\dots,\frac{\sum_{i\in \pi^{-1}(n)} z_i}{p(n)}}. \]
	Now let $T'$ be a set of the same size as $T$ such that $\pi(T')=T$; that is, for each element $i\in T$ we choose an arbitrary representative of $\pi^{-1}(i)$ and call the collection of these representatives $T'$. What is the chance that $S'\sim \mu'$ is contained in $T'$? To produce $S'$, we can first sample $S\sim \mu$. It must be that $S\in \binom{T}{k}$, or else $S'$ will not be contained in $T'$; but this is not enough. For each such set $S$, the chance that the random replacements of all elements end up coinciding with our representative choices for $T'$ is exactly $\prod_{i\in S} p(i)^{-1}$. So we get that
	\[ \PrX{S'\sim \mu'}{S'\subseteq T'}=\sum_{S\in \binom{T}{k}}\parens*{\mu(S)\prod_{i\in S}p(i)^{-1}}. \]
	On the other hand, the marginals of $\mu'$ are just $1/p(i)$ fraction of those of $\mu$. That is
	\[ \PrX{S'\sim \mu'}{i\in S'}=\frac{\PrX{S\sim \mu}{\pi(i)\in S}}{p(\pi(i))}. \]
	So
	\[ \sum_{i\in T'} \frac{\PrX{S'\sim \mu'}{i\in S'}}{k}=\sum_{i\in T}\frac{\PrX{S\sim \mu}{i\in S}}{kp(i)}. \]
	The desired inequality follows from applying \cref{lem:neg2} to $\mu'$.
\end{proof}

	\section{Hierarchical Sampling}
\label{sec:hierarchical}

In this section we complete the proof of \cref{thm:main} by
presenting an efficient implementation of the Markov chain
$\Mc_{\mu,p}^t$ defined in \cref{sec:instant}. In the next section we show how to
compute the sampling distribution $p$ so that it
approximates the marginals of $\mu$ sufficiently well, as required by
\cref{lem:coupling}. We measure the complexity by the number of
queries of $\mu(S)$, since they represent the dominant cost of all the
algorithms. The main tool we use to accomplish both of
these tasks is the ``down-up'' Markov chain, denoted here as
$\Mc_\mu$, which for any $\mu$ with a log-concave generating
polynomial converges to within $\epsilon$ total variation distance of $\mu$ after $O(k\log(k/\epsilon))$
steps, and each step requires $n$ queries (\cref{thm:down-up}).

Our implementation of $\Mc_{\mu,p}^t$ can be viewed as a hierarchical
sampler with two levels: the top level is the chain $\Mc_{\mu,p}^t$
itself, whereas the bottom level is a separate MCMC sampler which is
used to approximate each step of the top level. The expensive part of
each step of $\Mc_{\mu,p}^t$ involves sampling from the distribution
$\mu_{\sigma,p}$ which is defined over $\binom{[t+k]}{k}$. Since
$\mu_{\sigma,p}$ has a log-concave generating polynomial (see
discussion in \cref{sec:negdep}), we can use
$\Mc_{\mu_{\sigma,p}}$ to approximately sample from it. The overall
implementation of one step of $\Mc_{\mu,p}^t$ given state $\Sc_i$, illustrated in
\Cref{fig:hierarchical}, proceeds as follows:
\begin{enumerate}
  \item Construct random sequence $\sigma$ by i.i.d.~sampling $\rho_1,...,\rho_t$ from $p$
    and permuting $\Sc_i$ into it.
  \item Return $\Sc_{i+1}=\Sc_{i,s}$ obtained from
    $s$ steps of $\Mc_{\mu_{\sigma,p}}$ started at $\Sc_{i,0}=\Sc_i$.
\end{enumerate}\vspace{3mm}
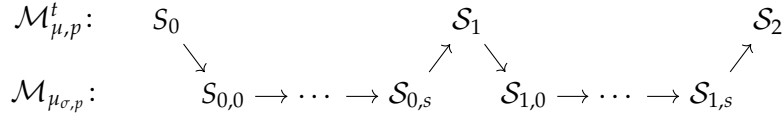
\begin{figure}
\centering
\begin{tikzpicture}
  \draw (0,0) node (A1) {$S_0$};
  \draw (0.75,-1) node (C1) {$S_{0,0}$};
  \draw (2,-1) node (C2) {$\dots$};
  \draw (3.25,-1) node (C3) {$\Sc_{0,s}$};
  \draw (4,0) node (A2) {$\Sc_1$};
  \draw (-1.5,0) node (D1) {$\Mc_{\mu,p}^t\!:$};
  \draw (-1.5,-1) node (D2) {$\Mc_{\mu_{\sigma,p}}\!:$};
  \draw [->] (A1) -- (C1);
  \draw [->] (C1) -- (C2);
  \draw [->] (C2) -- (C3);
  \draw [->] (C3) -- (A2);
  \draw (4.75,-1) node (C4) {$\Sc_{1,0}$};
  \draw (6,-1) node (C5) {$\dots$};
  \draw (7.25,-1) node (C6) {$\Sc_{1,s}$};
  \draw (8,0) node (A3) {$\Sc_2$};
  \draw [->] (A2) -- (C4);
  \draw [->] (C4) -- (C5);
  \draw [->] (C5) -- (C6);
  \draw [->] (C6) -- (A3);    
  \end{tikzpicture}
  \caption{Hierarchical sampling which uses $\Mc_{\mu_{\sigma,p}}$ to
    approximate each step of $\Mc_{\mu,p}^t$.}\label{fig:hierarchical}
\end{figure}

\begin{lemma}\label{lem:hierarchical}
  Let $\mu$ be a $k$-homogeneous distribution with a
log-concave polynomial and take any $0<\epsilon\leq 1/2$. If $s=
\lceil3k^2\log(1/\epsilon)\rceil$, $t= 20k^2$ and $k\cdot p(i)\geq 
  \big(1+\frac1{20k}\big)^{-1}\Prob_{\Sc\sim\mu}[i\in \Sc]$ for all $i\in[n]$, then
  starting from any $S_0\in\supp(\mu)$, the above procedure returns
  $\Sc_l$ within $\epsilon$ total variation distance of $\mu$ for $l= \lceil\log(2/\epsilon)\rceil$ and
  requires $\poly(k)\cdot O(\log^21/\epsilon)$ queries.
\end{lemma}
\begin{proof}
  The proof is again based on a coupling argument. Let $\Sc_i'$ denote
  the sequence of random states produced by the chain $\Mc_{\mu,p}^t$
  started at $S_0$. From \cref{lem:coupling} it follows that
  $\Sc_l'$ is within $\epsilon/2$ total variation distance of
  $\mu$. To couple this chain with our sampling procedure note that
since $s\geq 3k\log(k/\epsilon)\geq k\log(k/\delta)$ for $\delta =
\frac\epsilon{2\log(2/\epsilon)}$,  log-concavity of the generating
polynomial of $\mu_{\sigma,p}$ together with \cref{thm:down-up} imply that
any given step of our procedure is only $\delta$ total variation
distance away from an exact step of the chain $\Mc_{\mu,p}^t$. So, a
union bound implies that with probability at least $1-l\delta\geq 1- 
\frac\epsilon2$ the entire procedure is identical to the procedure
that generates $\Sc_l'$. We conclude that $\Sc_l$ must be within $\epsilon$ total
variation distance from $\mu$. Since the support of each distribution
$\mu_{\sigma,p}$ is contained in $2^{[20k^2+k]}$, the number of
queries per step of $\Mc_{\mu_{\sigma,p}}$ is bounded by $21k^2$.
Further, the total number of those steps is $ls\leq
3k\log(k/\epsilon)\log(2/\epsilon)$, so the number of queries is
$O(k^3\log^2k/\epsilon)$.
\end{proof}

	\section{Marginal Estimation and High-Precision Counting}
\label{sec:counting}

In this section we prove \cref{thm:oracle-construct,thm:counting}. We use the standard trick of introducing a cooling schedule. Our strategy is to introduce a sequence of distributions $\mu_0,\dots, \mu_t=\mu$ that together define a cooling schedule. All $\mu_i$ will have a log-concave generating polynomial. We will guarantee that $\mu_i$ and $\mu_{i+1}$ do not differ by more than a factor of $1+O(1/k)$ pointwise. We will make sure that $\mu_0$'s marginals are easy to estimate, and $t=\poly(\log n, k)$ is not too large. Then we use this cooling schedule to successively estimate the marginals of each $\mu_i$. This will prove \cref{thm:oracle-construct}. Afterwards we use standard unbiased estimators for the ratios of partition functions of successive $\mu_i$ and combine the estimates to prove \cref{thm:counting}. 

Our cooling schedule construction will be based on one initial set $U\in \binom{[n]}{k}$. We call $U$ admissible if $\PrX{\mu}{U}\geq \frac{1}{2\binom{n}{k}}$. Note that a random $U\sim \mu$ is guaranteed to be admissible with probability at least $1/2$:
\[ \sum_{U\text{ not admissible}} \PrX{\mu}{U}\leq \sum_{U\in \binom{[n]}{k}} \frac{1}{2\binom{n}{k}}\leq  \frac{1}{2}. \]
\begin{proposition}
	There is a randomized algorithm that outputs an admissible set with probability $1-\delta$ in time $O(nk\log(k)\log(1/\delta))$.
\end{proposition}
\begin{proof}
	We can use the naive down-up random walk of \cref{thm:down-up} to approximately sample $\log(1/\delta)$ many points from $\mu$. We then simply return the sample $U$ with the highest $\mu(U)$. Each sample takes time $O(nk\log(k))$ to produce, and the chance that none of them are  admissible is $2^{-\log(1/\delta)}\leq \delta$.
\end{proof}
Once we have found the admissible $U$, we let $\mu_i$ be defined as follows:
\[ \mu_i(S)=\lambda_i^{\card{S\cap U}}\mu(S), \]
where $\lambda_0\geq\lambda_1\geq \dots\geq \lambda_t=1$ define the cooling schedule. Note that each $\mu_i$ has a log-concave generating polynomial, because $g_{\mu_i}$ can be obtained from $g_\mu$ by scaling variables $z_e$ for $e\in U$ by $\lambda_i$, a linear transformation.

We will guarantee that $\lambda_{i-1}\leq  (1+O(1/k^2))\cdot \lambda_{i}$. This implies that the distributions $\mu_i$ and $\mu_{i-1}$ do not differ by more than a factor of $1+O(1/k)$ pointwise, because
\[ \frac{\mu_{i-1}(S)}{\mu_{i}(S)}\leq \parens*{\frac{\lambda_{i-1}}{\lambda_i}}^k=(1+O(1/k^2))^k=1+O(1/k). \]
Our goal is for $\lambda_t=1$, so that $\mu_t=\mu$. On the other hand, we would like $\mu_0$ to be a distribution whose marginals are easy to estimate, and for this, we would like $\mu_0$ to be very much close to a point mass distribution on the set $U$. This can be achieved by setting $\lambda_0$ to a very large value, but the tradeoff is that we need $t=k^2 O(\log(\lambda_0))$ many steps in our cooling process. A happy middle ground is $\lambda_0=n^{O(k)}$. Starting with this $\lambda_0$, we only need $t=k^2 \log(n^{O(k)})=k^3 \log n=\poly(k, \log n)$ many cooling steps. This choice also leads to easy bounds on the marginals of $\mu_0$:
\begin{proposition}
	Assuming that $U$ is admissible, the total variation distance between $\mu_0$ and the point mass distribution $\1_{U}$ is $\leq 2\binom{n}{k}/\lambda_0$.
\end{proposition}
\begin{proof}
	It is easy to see that 
	\[ \norm{\1_U-\mu_0}_\tv = 1-\PrX{\mu_0}{U}\leq \frac{\sum_{S\neq U} \mu(S)\lambda_0^{\card{S\cap U}}}{\mu(U) \lambda_0^k+ \sum_{S\neq U} \mu(S)\lambda_0^{\card{S\cap U}} }\leq \frac{\lambda_0^{k-1}\sum_{S} \mu(S) }{\lambda_0^k \mu(U)}\leq \frac{1}{\lambda_0\PrX{\mu}{U}}. \]
	and the desired inequality follows from admissibility of $U$.
\end{proof}
So whilst keeping $\lambda_0=n^{O(k)}$, we can make sure $\mu_0$ is inverse-polynomially close to $\1_U$. This allows us to easily construct a sampling distribution $p:[n]\to \R_{\geq 0}$ that satisfies the assumptions of \cref{thm:main} for $\mu_0$.
\begin{proposition}
	If $\norm{\1_U-\mu_0}_\tv\leq 1/n$, then the following distribution $p:[n]\to \R_{\geq 0}$ satisfies the assumptions of \cref{thm:main} for $\mu_0$:
	\[ 
		p(i):=\begin{cases}
			\frac{1}{k+1} & i\notin U,\\
			\frac{1}{(k+1)(n-k)} & i\in U.\\
		\end{cases}
	\]
\end{proposition}
\begin{proof}
	For $i\in U$, we trivially have
	\[ p(i)=\frac{1}{k+1}=\frac{1}{k+O(1)}\geq \frac{\PrX{S\sim \mu_0}{i\in S}}{k+O(1)}. \]
	For $i\notin U$, by the bound on total variation distance, we have $\PrX{S\sim \mu_0}{i\in S}\leq 1/n$. This implies
	\[ p(i)=\frac{1}{(k+1)(n-k)}\geq \frac{1}{n(k+1)}\geq \frac{\PrX{S\sim \mu_0}{i\in S}}{k+O(1)}. \]
\end{proof}

We are now ready to finish the proof of \cref{thm:oracle-construct}.
\begin{proof}[Proof of \cref{thm:oracle-construct}]
	Having constructed a sampling distribution $p$ for $\mu_0$, we can apply \cref{thm:main} and use samples to get an even more precise estimate of the marginals. We then use the updated marginal estimates for $\mu_1$. Since $\mu_0$ and $\mu_1$ differ by at most a factor of $1+O(1/k)$, our marginal estimates are valid for \cref{thm:main} to be applied to $\mu_1$. However we cannot keep going forward, or else we accumulate error. Instead we use \cref{thm:main} repeatedly to sample and from these samples extract a fresh good quality estimate of the marginals of $\mu_1$. We then continue the same procedure for $\mu_2$, $\mu_3$, and so on.
	
	All that we need to show is how to leverage samples produced by \cref{thm:main} to produce a new valid estimate $p$ of marginals whose quality is independent of what is fed to \cref{thm:main}. This is the contents of \cref{lem:estimation}. Setting $\epsilon=O(1/k)$ in \cref{lem:estimation}, we use $n\poly(k, \log n)$ many samples, each obtained in $\poly(\log n, k)$ time, to produce a fresh distribution $p$ satisfying the assumptions of $\cref{thm:main}$.
\end{proof}

\begin{lemma}\label{lem:estimation}
  Let $\mu: \binom{[n]}{k}\to \R_{\geq 0}$ be a distribution with a log-concave generating polynomial. For any $0<\epsilon<1$, it takes
  $O(nk^{-1}\epsilon^{-3}\log (n/\delta))$ samples from $\mu$ to generate a distribution $p:[n]\rightarrow \R_{>0}$ such that with probability
  $1-\delta$ we have $k\cdot p(i)\geq (1-\epsilon)\Prob_{\Sc\sim\mu}[i\in
  \Sc]$ for all $i\in[n]$.
\end{lemma}

\begin{proof}
Throughout the proof, let $q_i$ denote the marginal
$\PrX{S\sim\mu}{i\in S}$. Since we only require a one-sided approximation of the marginals, it
is acceptable that we significantly over-estimate those that are
sufficiently small. To that end, we define $T=\set{i\in[n] \given
q_i\geq \frac{\epsilon k}{4n}}$ as the set of marginals
which are ``large''. Let $S_1,...,S_s$ be sampled i.i.d.\ from
$\mu$ with $s=\lceil3\cdot4^3nk^{-1}\epsilon^{-3}\ln(2n/\delta) \rceil$.
Let $\hat q_i =\frac1s\sum_{j=1}^s\1_{[i\in S_j]}$ and define
$\hat T= \set{i\in[n] \given \hat q_i\geq \frac{\epsilon k}{3n}}$.
We construct our distribution $p$ as follows:
\begin{enumerate}
\item For every $i\in \hat T$, we let
  $p(i)=\frac1k(1-\frac34\epsilon)\hat q_i$.
\item For all remaining $i$, we let $p(i)= (1-\sum_{i'\in\hat
    T}p(i'))\frac1{|[n]\backslash\hat T|}$.
\end{enumerate}
Note that $\sum_ip(i)=1$. Furthermore, it follows from
\cref{lem:chernoff} that for any $i\in T$ we have: 
\begin{align*}
  \Pr*{\abs*{\hat q_i - q_i}\geq
  \frac\epsilon4\cdot q_i}
  \leq 2e^{-\epsilon^2sq_i/(3\cdot 4^2)}\leq
  2e^{-\log(\frac{2n}\delta)\frac{4n}{\epsilon k}q_i}
  \leq\frac\delta n.
\end{align*}
Thus, a union bound implies that with probability $1-\delta$ we have
$|\hat q_i - q_i|\leq\frac\epsilon4 q_i$ for all $i\in T$. From
now on, condition on this event. First, it implies that $\hat T\subseteq
T$, since if $i\in\hat T$ then $q_i\geq(1+\frac14)^{-1}\hat q_i\geq
\frac45\cdot\frac{\epsilon k}{3n}\geq\frac{\epsilon k}{4n}$. This ensures that for all $i\in\hat T$:
\begin{align*}
 k\cdot p(i) = (1-\tfrac34\epsilon)\hat q_i\geq
(1-\tfrac34\epsilon)(1-\tfrac14\epsilon)
  q_i\geq (1-\epsilon)q_i.
\end{align*}
To lower-bound the remaining $p(i)$, we upper bound the total
probability mass of $p$ in the set $\hat T$:
\begin{align*}
\sum_{i\in \hat T}p(i) \leq
  (1-\tfrac34\epsilon)(1+\tfrac14\epsilon)\frac1k\sum_{i\in \hat T}\mu_i\leq 1-\frac\epsilon2.
\end{align*}
Therefore, if $i\not\in\hat T$ then $k\cdot p(i)\geq \frac{\epsilon
  k}{2n}\geq \frac32\hat q_i$. Furthemore, either $i\in T$, in which
case $q_i\leq\frac 43\hat q_i\leq k\cdot p(i)$, or $i\not\in T$, and
then $q_i\leq \frac{\epsilon k}{4n}\leq k\cdot p(i)$.
\end{proof}

Next we prove \cref{thm:counting}.
\begin{proof}[Proof of \cref{thm:counting}]
	By using \cref{thm:oracle-construct}, we can prepare marginal estimates sufficient for \cref{thm:main} for each one of the distributions $\mu_i$. The total running time is $n \poly(\log n, k, \log(1/\delta))$ for a success probability of $1-\delta/\poly(n)$. So from now on we assume that we can produce a sample from each $\mu_i$ in time $\poly(k, \log n)$.
	
	Let $Z_i:=\sum_{S} \mu_i(S)$ be the partition function for $\mu_i$. We use the standard trick of writing a telescoping product
	\[ Z_t=\frac{Z_t}{Z_{t-1}}\cdot \frac{Z_{t-1}}{Z_{t-2}}\cdots \frac{Z_1}{Z_0}\cdot Z_0, \]
	and estimating each fraction individually. In order for the total multiplicative error to be $1+O(\epsilon)$, we need to make sure each factor gets approximated within a factor of $1+O(\epsilon/t)$. We will show that this part takes only $1/\epsilon^2 \cdot \poly(\log n, k, \log(1/\delta))$ time. To approximate $Z_{i+1}/Z_i$, we can use empirical averages of an unbiased estimator. If $S\sim \mu_i$, then $\mu_{i+1}(S)/\mu_i(S)$ becomes an unbiased estimator for $Z_{i+1}/Z_i$. Because we made sure $\mu_i$ and $\mu_{i+1}$ are not too different pointwise, the range of this unbiased estimator is $1\pm O(1/k)$. So taking an empirical mean of $(t/k\epsilon)^2\log(1/\delta)$ many such samples yields a $1+O(\epsilon/t)$-approximation of $Z_{i+1}/Z_i$ with probability $1-\delta/\poly(n)$. By \cref{thm:main}, each sample can be produced in time $\poly(k, \log n, \log(1/\delta))$ (where we push the total variation distance into the failure probability $\delta$, to assume our samples yield unbiased estimators).
	
	It remains to estimate $Z_0$. Note that $\mu_0(U)$ is already a very good estimate of $Z_0$, because we can make sure $\lambda_0$ is set such that $\PrX{\mu_0}{U}\geq 1-1/\poly(n)$. However in rare cases where the desired accuracy $\epsilon$ is smaller than this $1/\poly(n)$, we can do the following procedure: We continue introducing distributions behind $\mu_0$, namely $\mu_{-1}$, $\mu_{-2}$, and so on, each with a larger and larger $\lambda_i$. We only need to go back far enough that $2\binom{[n]}{k}/\lambda_{-i}\ll \epsilon/t$. This happens at $i=\poly(\log n, k)$. We then use $\mu_{-i}(U)$ as our estimate for $Z_{-i}$ and as before estimate the ratios $Z_{-(i-1)}/Z_{-i}$ and so on using empirical means of unbiased estimators. Note that to sample from $\mu_{-1}, \mu_{-2}, \dots$, we do not need to use new sampling distributions $p$ applicable to \cref{thm:main}. The one we have precomputed for $\mu_0$ works for all of them.
\end{proof}

        \printbibliography
\end{document}